\documentclass[11pt]{article}
\usepackage[utf8]{inputenc}
\usepackage{mleftright}
\mleftright

\usepackage{style,qcircuit}
\let\epsilon=\varepsilon

\usepackage{algorithm,algpseudocode,array}
\algrenewcommand\algorithmicrequire{\textbf{Input:}}
\algrenewcommand\algorithmicensure{\textbf{Output:}}

\newcommand{\wid}{0.33}

\title{A Quantum Algorithm for Functions of Multiple Commuting Hermitian Matrices}
\makeatletter
\renewcommand\@date{{
  \vspace{-\baselineskip}
  \large\centering
  
    \hspace*{-0.65cm}
    \begin{tabular}{p{\wid\textwidth}<{\centering} p{\wid\textwidth}<{\centering} p{\wid\textwidth}<{\centering}}
        Yonah Borns-Weil & Tahsin Saffat & Zachary Stier \\
        {\footnotesize \texttt{yonah\_borns-weil@berkeley.edu}} & {\footnotesize \texttt{tahsin\_saffat@berkeley.edu}} & {\footnotesize \texttt{zstier@berkeley.edu}}
    \end{tabular}
    
  \bigskip
  
  UC Berkeley Department of Mathematics\phantom{\textsuperscript{2}}
  
  \bigskip

  February 2023
}}
\makeatother

\begin{document}

\maketitle

\begin{abstract}
    Quantum signal processing allows for quantum eigenvalue transformation with Hermitian matrices, in which each eigenspace component of an input vector gets transformed according to its eigenvalue. In this work, we introduce the {\em multivariate quantum eigenvalue transformation} for functions of commuting Hermitian matrices. We then present a framework for working with polynomial matrix functions in which we may solve MQET, and give the application of computing functions of normal matrices using a quantum computer. 
\end{abstract}

\section{Introduction}\label{intro}

One fundamental matrix problem is the computation of matrix-valued functions. 
Specifically, we define the quantum eigenvalue transformation as follows:

\begin{center}
\begin{tabular}{p{0.85\textwidth}}
    \textbf{Quantum Eigenvalue Transformation (QET)}: {\em Let $A$ be a diagonalizable complex $N \times N$  matrix with eigenvalues $\{\lambda_k\}\subset\bD$ and corresponding eigenvectors $\lcr{\bv_k}$. Given a function $f:\C \to \C$, a state $\ket{\bv}=\sum c_k \ket{\bv_k}$, and $\eps>0$, output (a scaling of) the state $f(A) \ket{\bv}:=\sum c_k f(\lambda_j) \ket{\bv_k}$ to precision $\eps$.}
\end{tabular}
\end{center}

(Here, $\bD$ is the closed unit disk in $\C$.) QET may be implemented using the quantum signal processing algorithm (QSP) for a degree-$D$ polynomial $f$ and a Hermitian matrix $A$ with constant ancilla overhead and $O(D)$ instances of a quantum circuit implementation of $A$, for instance as a {\em block-encoding}, which comes with some prescribed additional ancillae and an intrinsic constant known as the {\em subnormalization factor} (see \secref{sec:circuit notation}); $f$ can be provided as a classical blackbox,\footnote{i.e., a blackbox from which we may query one input at a time. In principle $f$ could take time polynomial in the problem parameters to evaluate, but since this is all done on the classical side we treat it as a constant-time blackbox.} and $\ket{\bv}$ can be given via some prepare oracle $U$ of which we have only one copy.\footnote{i.e., we initialize all registers to a tensor of $\ket{0}$s and on $U$'s register apply it to obtain the input state $\ket{\bv}$, while $U$ acts on other states unpredictably.} As we will need $\ket{\bv}$'s prepare oracle exactly once and $f$'s blackbox is only accessed on the classical side, we are more concerned with minimizing the number of instances of $A$'s block-encoding. The QSP method requires that $f$ is $L^\infty$-bounded. 

QET is not well-understood for other classes of matrices. We propose the following generalization of QET that makes sense for several commuting matrices. One generic setting in which this would be of use is in the case of multiple commuting observables; it may be the case that one wishes to compute some function of several such observables, without having to compute individual measurements. 

In preparation for the definition, we remark that commuting matrices have a simultaneous eigenbasis, in particular commuting Hermitians have a simultaneous orthonormal eigenbasis. 

\begin{center}
\begin{tabular}{p{0.85\textwidth}}
    \textbf{Multivariate Quantum Eigenvalue Transformation (MQET)}: {\em Let $\bbA=\{A_{\ell}\}_{0\le\ell\le r}$ be a family of $r+1$ pairwise commuting diagonalizable complex $N \times N$  matrices. Let $\{\bv_k\}$ be a mutual eigenbasis of $\bbA$ and let $\bgl_k=(\lambda_{{\ell},k})_{0\le\ell\le r}\in\bD^{r+1}$ denote the vector of eigenvalues for the mutual eigenvector $\bv_k$. Given a function $f:\C^{r+1} \to \C$, a state $\ket{\bv}=\sum c_k \ket{\bv_k}$, and $\eps>0$, output (a scaling of) the state $f(\bbA) \ket{\bv}:=\sum c_k f(\bgl_k) \ket{\bv_k}$ to precision $\eps$.}
\end{tabular}
\end{center}

The matrices $A_\ell$ can be provided as block-encodings, the function $f:\C^{r+1}\to\C$ can be provided as a blackbox, and the state $\ket{\bv}$ can be provided by a single prepare oracle, all as in QET. 

To motivate our method for MQET, first consider the problem of QET for normal matrices. Suppose we want to compute QET for a normal matrix $M$ and function $f:\C \rightarrow \C$ that is $L^\infty$-bounded on the unit square. We retain the boundedness assumption from the case where $M$ is Hermitian because it occurs as a special case of the problem we are currently considering. Let us also make the additional assumption that $f$ is a polynomial. This is a mild assumption as any $L^\infty$-bounded function on the unit square can be approximated by $L^\infty$-bounded polynomials. Now, make the critical observation that any normal matrix $M$ decomposes as $A+iB$, where $A$ and $B$ are {\em commuting} Hermitian matrices. We look for a decompostion 
$$f(A+iB)\approx\sum_k P_k(A)Q_k(B)$$
for some polynomials $P_k,Q_k$. Then, we could implement $P_k(A)$ and $Q_k(B)$ using QSP and combine the products using the linear combination of unitaries primitive. The number of block-encoding instances required grows with the number of terms in the sum, and the subnormalization factor grows with the $L^\infty$ norms of the polynomials $P_k,Q_k$ appearing in the decomposition. A reasonable idea would be to let $P_k(x)=x^k$. Then, the number of terms in the sum is $1+ \deg f$, but the $L^\infty$ norm of $Q_k$ blows up rapidly. However, by letting $P_k(x)=T_k(x)$, the Chebyshev polynomial of degree $k$, we can obtain a decomposition where the number of terms in the sum is still $1+\deg f$, but $Q_k$ remains $L^\infty$-bounded. It is not clear in general how to compute the optimal decomposition of the polynomial $f$. 

The method just outlined for QET of normal matrices inspires our method for MQET of commuting Hermitian matrices. Observe that $f(A+iB)=g(A,B)$ where $g:[-1,1]^2 \rightarrow \C$ is a bivariate $L^\infty$-bounded polynomial. In general, we may compute MQET for a family $\bbA=\{A_{\ell}\}_{0\le\ell\le r}$ of pairwise commuting Hermitian matrices and $f:[-1,1]^{r+1} \rightarrow \C$ as follows. Compute a decomposition 
$$f(\bbA)\approx\sum_\bk P_\bk^{(0)}(A_0)P_\bk^{(1)}(A_1)\cdots P_\bk^{(r)}(A_r)$$
where $\bk$ is indexed by tuples $(a_i)_{0\le i<r}$ and $P_\bk^{(i)}=T_{a_i}$ for $0\le i<r$. The number of terms scales with the degree vector of $f$ and the the polynomials appearing in the decomposition remain $L^\infty$-bounded. 

We believe that this framework is promising and will have many applications beyond those provided in this work. 

Let $1+\deg f=D=2^d$. We present the following: 
\begin{theorem}
    Consider any pairwise commuting Hermitians $\bbA=\lcr{A_\ell}_{0\le\ell\le r}$. Given block-encodings $U_{A_\ell}\in\BE_{\ga,m}(A_\ell;\eps)$, target precision $\eps'>0$, and an ($r+1$)-variate polynomial $g$ approximating a given continuous $L^\infty$-bounded function $f:[-1,1]^{r+1}\to\C$ as $\norm{f-g}_{L^\infty[-1,1]^{r+1}}\in O(\eps')$, there is an algorithm (\href{alg:cap}{Algorithm 1}) for implementing a block-encoding of $f(\bbA)$ using $O(rD^{r+1})$ instances of the unitaries $U_{A_\ell}$, $O(r(m+d))$ additional ancillae, and $O\lpr{(D+2)^r}$ subnormalization factor. Applying this matrix to $\ket{\bv}$ solves the given MQET instance with precision $\eps'$. 
\end{theorem}
As a consequence, we may specialize to the case of bivariate functions to obtain the following result about QET for normal matrices: 
\begin{theorem}
    Consider any normal matrix $M$. Given block-encoding $U_M\in\BE_{\ga,m}(M;\eps)$, target precision $\eps'>0$, and a bivariate polynomial $g$ approximating a given $L^\infty$-bounded function $f:\C\to\C$ as $\norm{f-g}_{L^\infty([-1,1]+i[-1,1])}\in O(\eps')$, there is an algorithm (\href{alg:mqet cap}{Algorithm 2}) for implementing a block-encoding of $f(M)$ using $O(D^2)$ instances of the unitary $U_M$, $O(m+d)$ additional ancillae, and $O\lpr{D}$ subnormalization factor. Applying this matrix to $\ket{\bv}$ solves the given QET instance with precision $\eps'$. 
\end{theorem}
A further application is to perform a sort of `Hamiltonian' simulation, where the `Hamiltonian' operator is actually just any normal matrix (not necessarily Hermitian). 

This framework fits into the recent development of quantum matrix functions, with hallmarks such as Grover's search algorithm \cite{Grover} and its generalization amplitude amplification, matrix inversion (Harrow--Hasidim--Lloyd \cite{HHL}), the linear combination of unitaries algorithm (Childs--Wiebe \cite{CW}, which we use as a subroutine), and quantum singular value transformation (Gily\'en--Su--Low--Wiebe \cite{GSLW}). Our results even more deeply rely on QSP of Low--Chuang \cite{LC}, and furnish approximate versions for both the specialized setting of normal operators as well as families of commuting operators, albeit with potentially steep costs and without so general a framework as offered by QSP and QSVT (see the {\em grand unification} of Martyn--Rossi--Tan--Chuang \cite{MRTC}). This is similar in spirit to the recent work of Rossi--Chuang \cite{RC} on {\em multivariable quantum signal processing} (M-QSP); many of the deliverables obtained in this paper also appear, possibly with different associated costs, in \cite{RC}. 

We briefly describe some differences between this work and \cite{RC}. M-QSP achieves extremely short block-encodings, using linearly-many instances of the block-encodings of the input matrices, but relying on a separate conjecture about the Fej\'er--Riesz theorem, on a relatively restricted class of functions, and requiring that the {\em block-encodings} commute---not just that the underlying, block-encoded matrices themselves commute. In dropping constraints on functions (and their polynomial approximations) and the structure of the block-encodings, we incur greater cost in block-encoding parameters and total instances used. 

Our application here to normal matrices recalls the aims of the recent work of Takahira--Ohashi--Sogabe--Usuda \cite{TOSU, TOSU 2} which uses quadrature and Cauchy's integral formula to implement matrix functions, though \cite{TOSU, TOSU 2}'s and \secref{sec:normal be}'s approaches are ultimately fundamentally incomparable. The idea in \cite{TOSU 2} is for any holomorphic function and any diagonalizable matrix to approximate the integral via Cauchy's integral formula using an appropriately-chosen contour containing the matrix's eigenvalues. In order to compare resulting block-encoding parameters, one must somehow contrast the parameters chosen for a given quadrature against a polynomial approximation. As the philosophies of the algorithms are simply incompatible, any comparison cannot hold, and indeed there are instances where the quadrature approach is advantageous (such as when the eigenvalues in $\bD$ are well-understood and lie in a convenient contour's interior), while there are also instances where the polynomial decomposition approach is advantageous (such as when the function is nearly an indicator function in $\bD$, which admits poor $L^\infty$ holomorphic approximation). 

Lastly, the work of Guo--Mitarai--Fujii \cite{GMF} concerns a quantum operation distinct from matrix functions, instead known as nonlinear transformation of complex amplitudes (NTCA), in which the individual amplitudes of a state (with respect to a canonical basis) are transformed by a polynomial. Our approach gives rise to a solution to this problem for a less constrained class of functions than those considered in \cite{GMF}. 

The paper is organized as follows. In \secref{sec:prelim} we describe present methods for implementing QET. This includes the framework of block-encoding and QSP for Hermitian matrices, as well as some newer literature on QET algorithms for other families of matrices. In \secref{sec:normal be} we describe the proposed method for normal matrices, revealed to be a special case of the algorithm for MQET (\secref{sec:general comm herm}). \secref{sec:apps} details an application of the algorithm for normal matrices to NTCA. 

\section{Preliminaries}\label{sec:prelim}

\subsection{Mathematics notation}

We will always consider operators on $n$ qubits (possibly with ancillae), and will put $N:=2^n$. In general, upper-case letters used for integers will be powers of 2, whose logarithms are given in lower-case. All logarithms are base 2. By $\U(N)$ we refer to the unitary linear operators on $\C^N$, by $\PU(N)$ we refer to the quotient of $\U(N)$ by scalars, and by $\Mat(N)$ we refer to all linear operators on $\C^N$, that is, $\C^{N\times N}$. $\bD$ is the {\em closed} unit disk in $\C$. For $d\in\N$, let $[[d]]:=\{0,\dots,d-1\}$. Let $\norm{\cdot}$ be the operator norm. Let $\norm{\cdot}_p$ be the $L^p$-norm on $[-1,1]$, $p\in[1,\infty]$. We use the shorthand ``$f:S\to\C$ is $L^\infty$-bounded'' (for any domain $S$) to say that $\norm{f}_{L^\infty S}\in O(1)$, independently of other ambient parameters (such as the number of qubits in the relevant circuit). Finally, $T_k$ is the $k$th Chebyshev polynomial.

\subsection{Circuit notation}\label{sec:circuit notation}

When we provide examples of quantum circuits, we will predominantly use standard notation, and will introduce two new notations intended for visual ease. 

The first is a compact notation for controls on ancilla registers. We say (consistent with the standard notation) that a single-wire control as in \figref{fig:controls} is {\em 0-controlled} on the left and {\em 1-controlled} on the right. 
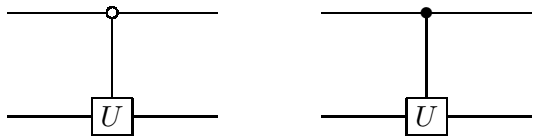
\begin{figure}
    \centering
    \leavevmode
    \Qcircuit @! {
    & \ctrlo{1} \qw & \qw & & \ctrl{1} \qw & \qw \\
    & \gate{U} & \qw & & \gate{U} & \qw
    }
    \caption{0- and 1-controlled $U$-gates.}
    \label{fig:controls}
\end{figure}
Suppose we have a $m$-ancilla register $A$, indexed from the top as $0,1,\dots,m-1$, and that $U$ acts on the register $B$ with a 0- or 1-control from each wire of $A$. We will compactly identify this with a number written in binary: the $j$-indexed wire is the $j$th place value, and it is 0 if the corresponding control is 0-controlled, and 1 if 1-controlled. If the resulting binary value is $b$, then this will be indicated by grouping all $m$ wires and putting $\ket{b}$ in a rounded control box. For instance, we have the equivalent circuits depicted in \figref{fig:controls ex}. 
\begin{figure}
    \centering
    \leavevmode
    \Qcircuit @C=1em @R=0.7em {
    \lstick{} & \ctrl{4} \qw & \ctrl{4} \qw & \qw \\
    \lstick{} & \ctrl{3} \qw & \ctrlo{3} \qw & \qw \\
    \lstick{} & \ctrl{2} \qw & \ctrl{2} \qw & \qw \\
    \lstick{} & \ctrl{1} \qw & \ctrl{1} \qw & \qw 
      \inputgroupv{1}{4}{0.8em}{1.4em}{A} \\
    \lstick{B} & \gate{U} & \gate{W} & \qw 
    }
    \qquad\qquad
    \Qcircuit @C=1.5em @R=3em {
    \lstick{A} & {/^4} \qw & \measure{\ket{15}} \qwx[1] \qw & \measure{\ket{13}} \qwx[1] \qw & \qw \\
    \lstick{B} & \qw & \gate{U} & \gate{W} & \qw
    }
    
    \caption{An example of the new notation for multiply-controlled gates.}
    \label{fig:controls ex}
\end{figure}
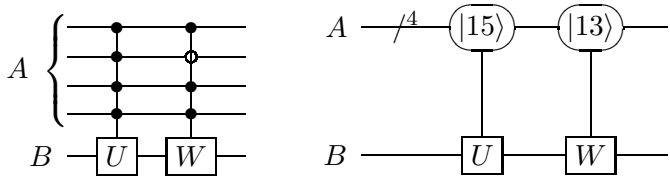

We also introduce notation for multi-wire gates that do not act on some wires that they are ``on top of.'' For instance, if there are three registers $A$, $B$, and $C$, which for broader clarity must be provided in that order, but a specific gate $U$ only acts on $A$ and $C$, then we will depict a ``break'' in the wire(s) for $B$. See for an example \figref{fig:skip registers}, which also depicts a formulation via SWAP gates which circumvents this notation, assuming that (say) $A$ and $B$ comprise the same number of qubits. Another example is \figref{fig:mult}. 
\begin{figure}
    \centering
    \leavevmode
    \Qcircuit @C=0.5cm @R=1em {
    \lstick{A} & \qw & \multigate{2}{U} & \qw & \qw \\
    \lstick{B} & \qw & \nghost{U} & & \qw \\
    \lstick{C} & \qw & \ghost{U} & \qw & \qw 
    }
    \qquad\qquad
    \Qcircuit @C=0.5cm @R=1em {
    \lstick{A} & \qswap \qwx[1] & \qw & \qswap \qwx[1] & \qw\\
    \lstick{B} & \qswap & \multigate{1}{U} & \qswap & \qw \\
    \lstick{C} & \qw & \ghost{U} & \qw & \qw 
    }
    \caption{An example of the new notation for gates that don't incorporate a specific register.}
    \label{fig:skip registers}
\end{figure}
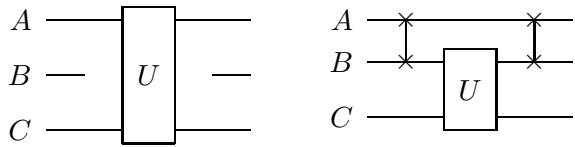

\subsection{Block-encodings} 

Computations with general matrices are most easily performed in the block-encoding model, which we summarize here. 

Intuitively, a block-encoding of a not necessarily unitary matrix $A$ is a large unitary matrix 
$$U_A=\begin{pmatrix}\frac{1}{\alpha}A&\star\\\star&\star\end{pmatrix}$$ 
where the $\star$'s represent arbitrary matrices (not necessarily all square). More specifically, given a $n$-qubit matrix $A$, we say that a ($n+a$)-qubit matrix $U_A$ is an {\em $(\alpha, a,\eps)$-block-encoding} of $A$ if 
$$\|A-\alpha(\bra{0^a}\otimes I_n)U_a(\ket{0^a}\otimes I_n)\|<\eps.
$$ 
This is also written as $U_A\in\BE_{\alpha, a}(A;\eps)$. $\ga$ is known as the {\em subnormalization factor} and is related to the number of repetitions needed to obtain the correct state after applying a block-encoding (since one wishes to measure the first register to be $\ket{0^a}$). 

To highlight its applicability, we note that if each matrix entry may be efficeintly computed then so too can a block-encoding, by the universality of the Toffoli gate in classical circuitry and the Toffoli+$H$ gates in quantum circuitry \cite{Shi,Aharonov}. Further, $(S,n+1,0)$-block-encodings are implementable using specialized oracles which identify the locations of nonzero entries in $A$, in the event that $A$ has at most $S$ nonzero entries in each row and in each column, and in the case that (e.g.) $A$ has exactly $S$ bands the same oracles give $(S,\log S+1,0)$-block-encodings \cite[\S6.3]{Lin}. However, for our purposes, we will fully blackbox the implementation details of our given block-encodings, as matrix-specific considerations quite often inform the choice of block-encoding model. 

Various functions are relatively easy (but not free) to compute with block-encoded matrices. We highlight two, which we will use later on. 

\begin{prop}[{block-encoding of product \cite[Lemma 53]{GSLW}}]\label{prop:be prod}
Given $n$-qubit matrices $A,B$ and $U_A\in\BE_{\alpha, a}(A;\eps_A)$, $U_B\in\BE_{\beta, b}(B;\eps_B)$, then $(I_b\otimes U_A)(I_a\otimes U_B)\in\BE_{\alpha\beta,a+b}(AB;\alpha\eps_B+\beta\eps_A)$, where $I_a, I_b$ act on the ancilla registers of $A,B$ respectively.\footnote{That is, suppose the work register is labeled 3, and we have an $a$-qubit register 1 and a $b$-qubit register 2. Then we arrange $U_A$ to be acting on 1 and 3, and $U_B$ to be acting on 2 and 3.}
\end{prop}

See also \figref{fig:mult} for a depiction of this circuit. 
\begin{figure}
    \centering
    \leavevmode
\Qcircuit @R=2em @C=0.7em {
\lstick{1} & \qw & {/^a} \qw & \qw & \qw & \qw & \qw & \multigate{2}{U_A} & \qw & & \nghost{U_B} & & & \qw & \qw \\
\lstick{2} & \qw & {/^b} \qw & \qw & \qw & \qw & & \nghost{U_A} & & & \ghost{U_B} & \qw & \qw & \qw & \qw \\
\lstick{3} & \qw & {/^n} \qw & \qw & \qw & \qw & \qw & \ghost{U_A} & \qw & \qw & \multigate{-2}{U_B} & \qw & \qw & \qw & \qw
    }
    \caption{\cite{GSLW}'s multiplication circuit.}
    \label{fig:mult}
\end{figure}
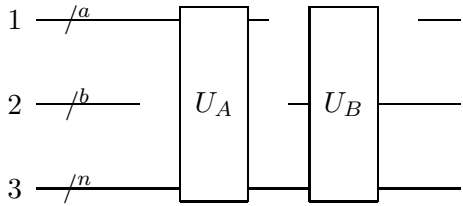

\begin{prop}[{block-encoding linear combinations of unitaries (LCU) \cite{CW} (see also \cite[\S7.3]{Lin})}]\label{prop:lcu}
Given access to a {\em prepare oracle} $V$ satisfying 
$$V\ket{0^n}=\frac{1}{\|\bba\|_1}\sum_{0\le j<K}\sqrt{a_j}\ket{j}$$
for $K=2^k$, and $U_0, \dots, U_{K-1}\in\U(N)$ arranged in a {\em select oracle} 
$$U:=\sum\limits_{0\le j<K}\ket{j}\bra{j}\otimes U_j,$$ 
then we have that
$$(V^\dag\otimes I_n)U(V\otimes I_n)\in\BE_{\norm{\bba}_1,k}\lpr{\sum\limits_{0\le j<K}a_jU_j;0}.$$
If $U$ and $V$ are implemented so that actually $(V^\dag\otimes I_n)U(V\otimes I_n)\in\BE_{\norm{\bba}_1,k}\lpr{\sum\limits_{0\le j<K}a_jU_j;\eps}$ then we will write $\LCU_\eps\lpr{\lpr{U_i}_{0\le i<K},\bba}$ for the circuit. 
\end{prop}

Using quantum signal processing (QSP) (for Hermitian matrices) or quantum singular value transform many other matrix functions can be implemented using polynomial approximations and block-encodings, cf.\ \cite[Chapters 7--8]{Lin}. See in particular the algorithm named {\em QET for Hermitian matrices} \cite[\S7.5]{Lin}. 

\section{QET for normal matrices} \label{sec:normal be}
    
    Recall the notion of block-encoding non-unitary matrices into larger, unitary matrices which may then be implemented as quantum circuit primitives. If $U_A\in\BE_{\ga,m}(A;\eps)$ is such a block-encoding and $f:\C^{N\times N}\to\C^{N\times N}$ is a matrix function, one seeks to implement a circuit for $U_{f(A)}\in\BE_{\ga',m'}(f(A);\eps')$ using only the standard quantum primitives as well as $U_A$ (and $U_A^\dag$), with as few calls to $U_A$ and $U_A^\dag$ as possible, $m'$ as small as possible (minimal number of ancillae), and $\ga'$ as small as possible (maximal success probability). For $A$ Hermitian, this problem is discussed extensively in \cite[Chapter 7]{Lin}. This analysis crucually depends on the $A$-block being Hermitian in order to construct an alternating sequence of $\C$-spaces, leaving the question of block-encoding matrix functions for {\em normal} matrices unresolved by these techniques. In this section we put forth an algorithm to implement quantum circuits for a given matrix function $f$ of a given normal matrix $M$. We assume that it is feasible to produce an arbitrarily good approximation to $f$ by a complex polynomial $g(z)=G(x,y)$ where $z=x+iy$ (see \secref{sec:poly approx}), into which we may plug in $z=M$.\footnote{The technical formulation of this condition is that $f$ on $\bD$ satisfies \propertyref{property:extensibility}, $(\g,k)$-extensibility for (say) $\g,k\in O(1)$.} Classically computing phase factors for QSP is a feasible task, which has been studied in Haah \cite{Haah} and Dong--Meng--Whaley--Lin \cite{DMWL}. However, at no point do we assume that computing $M$'s eigenvalues is tractable. 

    For clarity, we recall that a normal matrix is one which commutes with its Hermitian adjoint, and normal matrices are precisely those which are unitarily diagonalizable. 
    
    We defer a treatment of the most general case to \secref{sec:general comm herm}. 
    
    \subsection{Algorithm}\label{sec:normal alg}

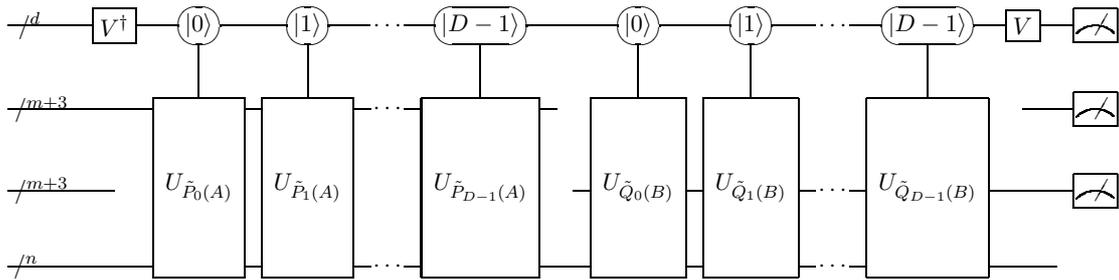
\begin{figure}
\centering
\leavevmode
\resizebox{\textwidth}{!}{
\Qcircuit @R=2em @C=0.7em {
& \qw & {/^{d\phantom{+3}}} \qw & \qw & \qw & \gate{V^\dag} & \measure{\ket{0}} \qwx[1] \qw & \measure{\ket{1}} \qwx[1] \qw & \qw & \cdots & & \measure{\ket{D-1}} \qwx[1] \qw & \qw & \qw & \measure{\ket{0}} \qwx[1] \qw & \measure{\ket{1}} \qwx[1] \qw & \qw & \cdots & & \measure{\ket{D-1}} \qwx[1] \qw & \gate{V} & \qw & \meter \\
& \qw & {/^{m+3}} \qw & \qw & \qw & \qw & \multigate{2}{U_{\tilde{P}_0(A)}} & \multigate{2}{U_{\tilde{P}_1(A)}} & \qw & \cdots & & \multigate{2}{U_{\tilde{P}_{D-1}(A)}} & \qw & & \nghost{U_{\tilde{Q}_0(B)}} & \nghost{U_{\tilde{Q}_1(B)}} & & & & \nghost{U_{\tilde{Q}_{D-1}(B)}} & & \qw & \meter \\
& \qw & {/^{m+3}} \qw & \qw & \qw & \qw & \nghost{U_{\tilde{P}_0(A)}} & \nghost{U_{\tilde{P}_1(A)}} & & & & \nghost{U_{\tilde{P}_{D-1}(A)}} & & & \ghost{U_{\tilde{Q}_0(B)}} & \ghost{U_{\tilde{Q}_1(B)}} & \qw & \cdots & & \ghost{U_{\tilde{Q}_{D-1}(B)}} & \qw & \qw & \meter \\
& \qw & {/^n{\phantom{+3}}} \qw & \qw & \qw & \qw & \ghost{U_{\tilde{P}_0(A)}} & \ghost{U_{\tilde{P}_1(A)}} & \qw & \cdots & & \ghost{U_{\tilde{P}_{D-1}(A)}} & \qw & \qw & \multigate{-2}{U_{\tilde{Q}_0(B)}} & \multigate{-2}{U_{\tilde{Q}_1(B)}} & \qw & \cdots & & \multigate{-2}{U_{\tilde{Q}_{D-1}(B)}} & \qw & \qw
}
}
\caption{Circuit for block-encoding of $f(M)$. See \secref{sec:mqet alg} for the algorithm, and \hyperref[alg:cap]{Algorithm 1} for a pseudocode description. $V$ comes from LCU (\propref{prop:lcu}), and the remaining notation on the top wire is introduced in \secref{sec:circuit notation}. We may let $P_k$ be $T_k$, as discussed in \secref{sec:subnorm}. }
\label{fig:normal circuit}
\end{figure}
        
        We describe an algorithm for obtaining a block-encoding of $f(M)$ given a block-encoding $U_M\in\BE_{\ga,m}(A;\eps)$ of a normal matrix $M$ and a function $f:\bD\to\bD$. We let $f$ be $(\g,k)$-extensible (\propertyref{property:extensibility}) for $\g,k\in O(1)$. This is a technical condition, defined in \secref{sec:poly approx}, which just ensures that $f$ is well-approximable. Let us assume that $f$ is well-approximated on $\bD$ by a polynomial $g$.
        
        Consider the decomposition
        \begin{equation}
            M=A+iB \label{eq:split}
        \end{equation}
        where $A$ and $B$ are commuting Hermitians with the same dimensions as $M$. It is not hard to obtain block-encodings of $A$ and $B$. In fact, all three are simultaneously diagonalizable, where if $M=U\cD U^\dag$ then $A=U(\re\cD)U^\dag$ and $B=U(\im\cD)U^\dag$. (It is this observation that powers \secref{sec:hamsim}.) In light of this, we immediately see how to implement block-encodings of $A$ and $B$: 
        \begin{align}
            U_A&=\LCU_\eps\lpr{\lpr{U_M,U_M^\dag},\lpr{\nhalf,\nhalf}}, & U_A&\in\BE_{\ga,m+1}(A;\eps) \nonumber \\
            U_B&=\LCU_\eps\lpr{\lpr{U_M,U_M^\dag},\lpr{\nicefrac{-i}{2},\nicefrac{i}{2}}}, & U_B&\in\BE_{\ga,m+1}(B;\eps). \nonumber 
        \end{align}
        As suggested in \secref{intro}, we can expand $g(x+iy)=G(x,y)$ in terms of polynomials in $x$ and $y$ so that 
        \begin{equation}
            g(M)=G(A,B)=\sum\limits_{0\le k<D}P_k(A)Q_k(B)\label{eq:vanilla sum-prod}
        \end{equation}
        for some $P_k,Q_k\in\C[z]$ of degree less than $D$ (where $g$ is also of degree less than $D$; suppose for convenience that $D$ is a power of 2). In the following, we describe how to use this decomposition, along with QSP, to obtain a block-encoding of $g(M)$, and how the parameters of the block-encoding depend on the decomposition we use. The implementation is also summarized in \figref{fig:normal circuit}. We will find that the subnormalization factor of the block-encoding of $g(M)$ is proportional to $\norm{\bgb}_1$, where $\bgb\in\R_{\ge0}^D$ is defined by $\gb_k :=\norm{P_k}_\infty\norm{Q_k}_\infty$. Our method for optimizing $\norm{\bgb}_1$ is to let $P_k=T_k$ and 
        $$Q_k=\frac{2^{1-\delta_{k,0}}}{\pi}\int_{-1}^1g(x+iy)T_k(x)\frac{\d x}{\sqrt{1-x^2}}.$$
        In \propref{prop:good scaling}, we show that this guarantees $\norm{\bgb}_1\le2D$. Our entire algorithm is summarized in \hyperref[alg:cap]{Algorithm 1}.
        
        Elaborating on how to use the decomposition in \eqref{eq:vanilla sum-prod} to implement a block-encoding of $g(M)$, first, normalize to $\tilde{P}_k=\frac{P_k}{\norm{P_k}_\infty}$ and $\tilde{Q}_k=\frac{Q_k}{\norm{Q_k}_\infty}$. \figref{fig:normal circuit} describes how to obtain $g(M)$ as quantum circuitry, we (classically) compute the phase factors $\Phi_k,\Psi_k$ for the polynomials $x\mapsto \tilde{P}_k(x),\tilde{Q}_k(x)$ (respectively), e.g.\ via \cite{DMWL}.\footnote{For clarity, we assume that one's implementation of phase factor computation is exact. However, error here is relatively insignificant: if each phase factor of $\Phi_k$ is off by $\eps'$ additively then we may crudely bound that the block-encoding of $\tilde{P}_k(A)$ is off by $D\eps'$, but since this computation is done classically it is sensible to assert that its error vanishes relative to the error on the quantum computer.} Then, we use QSP to compute $U_{\tilde{P}_k(A)}\in\BE_{\ga,m+2}(\tilde{P}_k(A);\eps)$ and $U_{\tilde{Q}_k(B)}\in\BE_{\ga,m+2}(\tilde{Q}_k(B);\eps)$, and compute via \cite[Lemma 53]{GSLW} the block-encoding of the product $U_{\tilde{P}_k(A)\tilde{Q}_k(B)}\in\BE_{\ga^2,2m+4}(\tilde{P}_k(A)\tilde{Q}_k(B);2\ga\eps)$. Finally, 
        \begin{align*}
            U_{g(M)}=\LCU_{2\ga\eps}\lpr{\lpr{\tilde{P}_k(A)\tilde{Q}_k(B)}_{0\le k<D},\bgb}, && U_{g(M)}\in\BE_{\ga^2\norm{\bgb}_1,2m+4+d}(g(M);2\ga\eps).
        \end{align*}
        
\begin{algorithm}
        \caption{block-encoding of $f(M)$}\label{alg:cap}
        \begin{algorithmic}[1]
        \Require $U_M\in\BE_{\ga,m}(M;\eps)$ for $M\in\Mat(N)$ normal; $g\in\C\lbr{z,\ol{z}}$ of degree less than $D$ such that $\norm{f-g}_\infty$ is small for $f:\bD\to\bD$
        \Ensure $U_{f(M)} \in \BE_{\ga',m'}(f(M);\eps')$
        \State $U_A\gets\LCU\lpr{\lpr{U_M,U_M^\dag},\lpr{\nhalf,\nhalf}}$ \Comment{$U_A\in\BE_{\ga,m+1}(A;\eps)$}
        \State $U_B\gets\LCU\lpr{\lpr{U_M,U_M^\dag},\lpr{\nicefrac{-i}{2},\nicefrac{i}{2}}}$ \Comment{$U_B\in\BE_{\ga,m+1}(B;\eps)$}
        \For{$0\le k<D$}
            \State
           $P_k \gets T_k$ \Comment{A Chebyshev polynomial}
          \State $Q_k(y)\gets\frac{2^{1-\delta_{k,0}}}{\pi}\int_{-1}^1g(x+iy)T_k(x)\frac{\d x}{\sqrt{1-x^2}}$
            \Comment{These satisfy $g(x+iy)=\sum\limits_{0\le k<D}T_k(x)Q_k(y)$}
            \State $\beta_k \gets \norm{P_k}_\infty \norm{Q_k}_\infty$
            \State
          $(\tilde{P}_k,\tilde{Q}_k) \gets (P_k/\norm{P_k}_\infty, Q_k/\norm{Q_k}_\infty) $
            \State Compute $U_{\tilde{P}_k(A)}\in\BE_{\ga,m+2}(\tilde{P}_k(A);\eps)$ \Comment{Using QSP}
            \State Compute $U_{\tilde{Q}_k(B)}\in\BE_{\ga,m+2}(\tilde{Q}_k(B);\eps)$ \Comment{Using QSP}
            \State Compute $S_k \in \BE_{\ga^2,2m+4}(\tilde{P}_k \tilde{Q}_k;2\ga\eps)$  
        \EndFor 
        \State \Return $\LCU\lpr{\lpr{S_k}_{0\le k<D},\bgb}$
        \end{algorithmic}
        \end{algorithm}

    \subsection{Analysis of the algorithm}
    
        In the previous section, we kept a running tally of the ancilla count; the final count was $2m+4+d$. We also track precision---$2\ga\eps$---and subnormalization factor---$\ga^2\norm{\bgb}_1$---above. 
        
        We can track the number of applications of $U_M$ (and its adjoint) as follows. $U_A$, $U_A^\dag$, $U_B$, and $U_B^\dag$ each requires two such gates; applying QSP for a degree less-than-$D$ polynomial multiplies the number of instances by at most $2D$. Computing products adds this cost. Running LCU adds these costs; the final LCU call is on $D$ unitaries. Thus there are at most $4(D-1)D\in O(D^2)$. Within the framework of decomposing a function by a polynomial of degree less than $D$ as in \eqref{eq:vanilla sum-prod} for which each monomial of the form $x^ky^{D-1-k}$ has nonzero coefficient (as is generically the case), the decomposition described using integration against the Chebyshevs uses as few $U_M$s as possible; see \secref{sec:qetn circuit depth}. 

        Observe that all nontrivial costs are incurred precisely when we deal with the many polynomials arising from the decomposition \eqref{eq:split}, by computing their paired products and then add using LCU (all of these inducing possibly large multiples of the base cost for each polynomial coming from QSP). 
    
    \subsection{Subnormalization factor}\label{sec:subnorm}
    
        \begin{proposition}\label{prop:good scaling}
            If $G(x,y)$ is degree less than $D$ in both $x$ and $y$ and satisfies $\norm{G}_{L^\infty[-1,1]^2}\le1$, then there are efficiently-computable degree-less-than-$D$ polynomials $P_k$ and $Q_k$ for $0\le k<D$ such that $G(x,y)=\sum\limits_{0\le k<D} P_k(x)Q_k(y)$ with $\norm{P_k}_{\infty} \le 1$ and $\norm{Q_k}_{\infty} \le 2$.
        \end{proposition}
        
        \begin{proof}
            Let $P_k:=T_k$ and 
            \begin{equation}
            \tilde{Q}_k(y):=\int_{-1}^1 G(x,y) T_k(x)\frac{\d x}{\sqrt{1-x^2}}.\label{eq:qtildedef}
            \end{equation}
            Let $Q_0=\frac{1}{\pi}\tilde{Q}_0$ and $Q_k=\frac{2}{\pi} \tilde{Q}_k$ for $k>0$. It is well-known that Chebyshev polynomials satisfy $\norm{T_k}_\infty=1$. Also, for $y \in [-1,1]$,
            $$|Q_k(y)| \le \frac{2}{\pi} \int_{-1}^1 |G(x,y)||T_k(x)|\frac{\d x}{\sqrt{1-x^2}} \le \frac{2}{\pi} \int_{-1}^1 \frac{\d x}{\sqrt{1-x^2}} =2.$$
            Chebyshev polynomials form an orthogonal set with respect to the measure $(1-x^2)^{-1} \d x$. Further, $T_k$ has norm $\pi$ for $k=0$ and $\nicefrac{\pi}{2}$ otherwise. It follows that by construction $G(x,y)=\sum\limits_{0\le k<D}P_k(x)Q_k(y)$.
            
            To compute $\tilde{Q}_k$, we use the techniques of \secref{sec:poly approx} and \secref{sec:cheb exp cheb int}: first write $G$ in the form \eqref{eq:p}. and evaluate \eqref{eq:qtildedef} as a sum using \propref{prop:int against cheb as sum}. 
        \end{proof}

        The algorithm described in \secref{sec:normal alg} implements a block-encoding of $f(M)$ with subnormalization factor bounded by $2\alpha D$ assuming that $\abs{f(x+iy)} \le 1$ for $x,y \in [-1,1]$. Because of the $L^\infty$ constraint on $f$ it is nontrivial to implement QET for several interesting functions. However, the approximation $\tilde{f}$ is not required to he holomorphic. This allows for a certain freedom. For example, if we wish to implement QET for the function $f(z)=z^n$, assuming that the eigenvalues are contained on the unit disk, then we actually need to choose a polynomial approximation $\tilde{f}$ with the following properties:
        
        \begin{itemize}
            \item $\abs{\tilde{f}(z)}\le 1$ for $\re z,\im z \in [-1,1]$
            \item $\abs{f(z)-\tilde{f}(z)}<\epsilon$ for $|z| \le 1$, for some approximation parameter $\epsilon$
        \end{itemize}
        
        Because the function $f(z)$ blows up outside of the unit disk, it seems these two conditions are in contention. Perhaps it is possible then that these two conditions can be simultaneously satisfied by a bivariate polynomial with reasonable degree. There are similar questions, for example, for functions such as $f(z)=\nicefrac{\kappa}{z}$ on the domain $\abs{z}>\kappa$. 

\section{On nonlinear transformation of complex amplitudes} \label{sec:apps}
We now demonstrate that the methods of \secref{sec:normal be} naturally provide an extension to the \emph{nonlinear transformation of complex amplitudes} as developed in \cite{GMF}. 

\begin{center}
\begin{tabular}{p{0.85\textwidth}}
    \textbf{Nonlinear Transformation of Complex Amplitudes (NTCA)}: {\em Let $F:\bD\to\C$ be any function. Given a state $\ket{\bv}=\sum c_k\ket{k}$ in the computational basis and $\eps>0$, output (a scaling of) the state $\sum F(c_k)\ket{k}$ to precision $\eps$. }
\end{tabular}
\end{center}

In the setting of \cite{GMF}, we have access to a (unitary) state preparation oracle $U$, as well as its adjoint and controlled versions, such that $U\ket{0}^{\otimes n}=\ket{\bv}$. \cite{GMF} addresses this problem with $\eps=0$ for 
\begin{equation}
    F(z)=F_1(z)=P(\re z)+Q(\im z) \label{eq:gmf fn type}
\end{equation}
for arbitrary continuous functions $P,Q:[-1,1]\to\C$. Specifically, if $P$ and $Q$ can be approximated up to error $\frac{\eps}{4N}$ by degree-$D$ polynomials $\tilde{P}$ and $\tilde{Q}$, respectively, such that $\tilde{F}_1(z):=\tilde{P}(\re z)+\tilde{Q}(\im z)$ approximates $F_1$, and $\g=\sup\limits_{x\in[-1,1]}\max\lcr{\abs{P(x)},\abs{Q(x)}},$ then\footnote{Suppose the normalization factor is $c:=\sum\limits_{0\le k<N}\abs{F(c_k)}^2$.} $\frac{1}{\sqrt{c}}\sum\limits_{0\le k<N}F_1(c_k)\ket{k}$ can be approximated to $L^2$ error $\nicefrac{\eps}{N}$, using in expectation
$$O\lpr{D\g\sqrt{\frac{N}{\sum\limits_{0\le k<N}\abs{\tilde{F}_1(c_k)}^2}}}$$ controlled $U$ and $U^{\dag}$ gates. 

The algorithm in \cite{GMF} uses the technique of \emph{block-encoding of amplitudes}. In particular, let $A=\lpr{I\otimes W^\dag}\tilde{G}(I\otimes W)$ and $B=\lpr{I\otimes W^\dag}\tilde{G}'(I\otimes W)$, where $W$, $\tilde{G}$, and $\tilde{G}'$ are defined as in \cite[Figures 1--3]{GMF}; for completeness, we also include their definitions and a brief explanation in \secref{sec:ntca defs}. Then $A$ and $B$ are block-encodings of matrices which map $\ket{k}$ to $(\re c_k)\ket{k}$ and $(\im c_k)\ket{k}$, respectively, that is, they are block-encodings of diagonal matrices containing the data of the original amplitudes. 

Applying the algorithm from \secref{sec:normal alg} to $A$ and $B$, we may obtain a block-encoding of a matrix that maps $\ket{k}$ to $F_2(c_k)\ket{k}$, for all $F_2$ which are $(\g,k)$-extensible, which is defined in \propertyref{property:extensibility}---allowing for an additional kind of function as those specified in \eqref{eq:gmf fn type}. Thus, this resolves NTCA with expected 
$$O\lpr{D^2\g\sqrt{\frac{N}{\sum\limits_{0\le k<N}\abs{F_2(c_k)}^2}}}$$ 
controlled $U$ and $U^{\dag}$ operations. A concrete instance of a newly-attainable function is 
\begin{align*}
    f(z)&=\frac{5}{1024}(z+\ol{z})^6(z-\ol{z})^4+\frac{17}{512}(z+\ol{z})(z-\ol{z})^8,\\
    \text{or }f(x+iy)&=5x^6y^4+17xy^8.
\end{align*}

\section{MQET for several commuting Hermitian matrices}\label{sec:general comm herm}

    We generalize the ideas in \secref{sec:normal be} following the observation that QET of normal matrices is a special case of MQET of two commuting Hermitian matrices.

    Consider a continuous function $f:[-1,1]^{r+1}\to\bD$ for which we select a polynomial approximation $g$, perhaps via the techniques of \secref{sec:poly approx}. Suppose the maximum total degree of any monomial in $g$ is less than $D$. 
    
    Generalizing the method from Section 3.1, we write down an expansion of $g$ analogous to \eqref{eq:vanilla sum-prod}. 
    Suppose we find polynomials $P_\bs^{(k)},Q_\bs\in\C[z]$ for $\bs$ ranging in $[[D]]^r$ and $k$ ranging from 0 to $r-1$ satisfying
    \begin{equation}
    g(\bx)=\sum\limits_{\bs\in[[D]]^r}Q_\bs(x_r)\prod\limits_{0\le k<r}P_\bs^{(k)}(x_k)
    \label{multivardecomp}
    \end{equation}
    \begin{remark}
        This sum has a factor of $D$ fewer terms than the number of monomials. This is crucial because we avoid applying LCU to each of the monomials, which causes the subnormalization factor to needlessly increase. 
    \end{remark}
    
    In the following subsection, we describe in detail how to use this decomposition, along with QSP, to obtain a block-encoding of $g(\bbA)$, and how the parameters of the block-encoding depend on the decomposition we use. The implementation is summarized in \figref{fig:mqet circuit}. We will find that the subnormalization factor of the block-encoding is proportional to $\norm{\bgb}_1$, where $\bgb\in\R_{\ge0}^{[[D+1]]^r}$ is defined by $\gb_\bs:=\norm{Q_\bs}_\infty\prod\limits_{0\le k<r}\norm{P_\bs^{(k)}}_\infty$. We optimize $\norm{\bgb}_1$ by letting $P_\bs^{(k)}=T_{s_k}$ (the $s_k$th Chebyshev polynomial) and defining $Q_\bs$ by an orthogonal projection formula proposed in \secref{sec:gen subnorm factor}; we show that this ensures $\norm{\bgb}_1\le2D^r$. Our entire algorithm is summarized in \hyperref[sec:mqet alg]{Algorithm 2}.

\subsection{Algorithm}\label{sec:mqet alg}

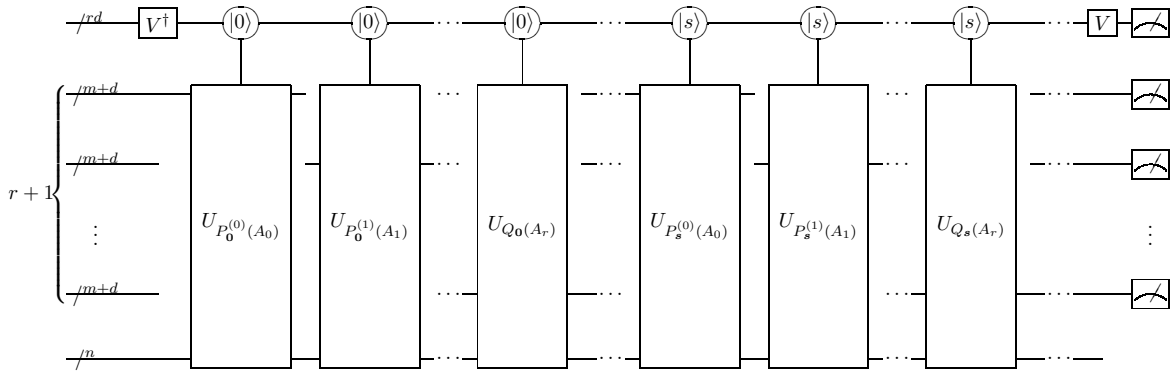
\begin{figure}
\centering
\leavevmode
\resizebox{\textwidth}{!}{
\Qcircuit @R=2em @C=0.7em {
& \qw & {/^{rd\phantom{+}}} \qw & \qw & \qw & \gate{V^\dag} & \measure{\ket{0}} \qwx[1] & \qw & \measure{\ket{0}} \qwx[1] & \qw & \cdots & & \measure{\ket{0}} \qwx[1] & \qw & \qw & \cdots & & \measure{\ket{s}} \qwx[1] & \qw & \measure{\ket{s}} \qwx[1] & \qw & \cdots & & \measure{\ket{s}} \qwx[1] & \qw & \qw & \cdots & & \gate{V} & \meter \\
& \qw & {/^{m+d}} \qw & \qw & \qw & \qw & \ghost{U_{P_\zero^{(0)}(A_0)}} & \qw & \nghost{U_{P_\zero^{(1)}(A_1)}} & & \cdots & & \nghost{U_{Q_{\zero}(A_r)}} & & \qw & \cdots & & \ghost{U_{P_\bs^{(0)}(A_0)}} & \qw & \nghost{U_{P_\bs^{(1)}(A_1)}} & & \cdots & & \nghost{U_{Q_{\bs}(A_r)}} & & \qw & \cdots & & \qw & \meter \\
& \qw & {/^{m+d}} \qw & \qw & \qw & \qw & \nghost{U_{P_\zero^{(0)}(A_0)}} & & \ghost{U_{P_\zero^{(1)}(A_1)}} & \qw & \cdots & & \nghost{U_{Q_{\zero}(A_r)}} & & \qw & \cdots & & \nghost{U_{P_\bs^{(0)}(A_0)}} & & \ghost{U_{P_\bs^{(1)}(A_1)}} & \qw & \cdots & & \nghost{U_{Q_{\bs}(A_r)}} & & \qw & \cdots & & \qw & \meter \\
& & \vdots & & & & \nghost{U_{P_\zero^{(0)}(A_0)}} & & \nghost{U_{P_\zero^{(1)}(A_1)}} & & & & \nghost{U_{Q_{\zero}(A_r)}} & & & & & \nghost{U_{P_\bs^{(0)}(A_0)}} & & \nghost{U_{P_\bs^{(1)}(A_1)}} & & & & \nghost{U_{Q_{\bs}(A_r)}} & & & & & & \vdots \\
& \qw & {/^{m+d}} \qw & \qw & \qw & \qw & \nghost{U_{P_\zero^{(0)}(A_0)}} & & \nghost{U_{P_\zero^{(1)}(A_1)}} & & \cdots & & \ghost{U_{Q_{\zero}(A_r)}} & \qw & \qw & \cdots & & \nghost{U_{P_\bs^{(0)}(A_0)}} & & \nghost{U_{P_\bs^{(1)}(A_1)}} & & \cdots & & \ghost{U_{Q_{\bs}(A_r)}} & \qw & \qw & \cdots & & \qw & \meter 
  \inputgroupv{2}{5}{0.8em}{4.8em}{r+1} \\
& \qw & {/^{n\phantom{+3}}} \qw & \qw & \qw & \qw & \multigate{-4}{U_{P_\zero^{(0)}(A_0)}} & \qw & \multigate{-4}{U_{P_\zero^{(1)}(A_1)}} & \qw & \cdots & & \multigate{-4}{U_{Q_{\zero}(A_r)}} & \qw & \qw & \cdots & & \multigate{-4}{U_{P_\bs^{(0)}(A_0)}} & \qw & \multigate{-4}{U_{P_\bs^{(1)}(A_1)}} & \qw & \cdots & & \multigate{-4}{U_{Q_{\bs}(A_r)}} & \qw & \qw & \cdots & & \qw 
}
}
\caption{Circuit for block-encoding of $f(A_0,\dots,A_r)$. Here we let $0\le s<D^r$ denote the integer given in base $D$ by $\bs$. See \secref{sec:mqet alg} for the algorithm, and \hyperref[alg:mqet cap]{Algorithm 2} for a pseudocode description. $V$ comes from LCU (\propref{prop:lcu}), and the remaining notation on the top wire is introduced in \secref{sec:circuit notation}. We may let $P_\bs^{(k)}$ be $T_{s_k}$, as discussed in \secref{sec:gen subnorm factor}. }
\label{fig:mqet circuit}
\end{figure}

        We describe an algorithm for obtaining a block-encoding of $g(\bbA)$ given block-encodings $\bU\in\prod\limits_{0\le k\le r}\BE_{\ga_k,m_k}\lpr{A_k;\eps_k}$, as well a decomposition of $g$ as in \eqref{multivardecomp}.
                
        For convenience, we introduce some notation. Let 
        $$T(\bs,\bbA):=Q_\bs(A_r)\prod\limits_{0\le k<r}P_\bs^{(k)}(A_k),$$
        so that $g(\bbA)=\sum\limits_{\bs\in[[D]]^r}T(\bs,\bbA)$. Let also 
        \begin{align}
            \ga:=\prod\limits_{0\le k\le r}\ga_k, &&
            m:=\sum\limits_{0\le k\le r}m_k, &&
            \eps:=\sum\limits_{0\le k\le r}\ga_k\eps_k.\label{eq:ga m eps}
        \end{align}
        
        Let $\bgb\in\R_{\ge0}^{[[D+1]]^r}$ be the vector where $\gb_\bs=\norm{Q_\bs}_\infty\prod\limits_{0\le k<r}\norm{P_\bs^{(k)}}_\infty$. Normalize to $\tilde{P}_\bs^{(k)}:=\frac{P_\bs^{(k)}}{\norm{P_\bs^{(k)}}_\infty}$ and $\tilde{Q}_\bs:=\frac{Q_\bs}{\norm{Q_\bs}_\infty}$. To obtain $g(\bbA)$ as quantum circuitry, we (classically) compute the phase factors $\Phi_\bs^{(k)},\Psi_\bs$ for the polynomials $x\mapsto\tilde{P}_\bs^{(k)}(x),\tilde{Q}_\bs(x)$ (respectively), e.g.\ via \cite{DMWL}. Then, we use QSP to compute $U_{\tilde{P}_\bs^{(k)}(A_k)}\in\BE_{\ga_k,m+2}\lpr{\tilde{P}_\bs^{(k)}(A_k);\eps_k}$ and $U_{\tilde{Q}_\bs(A_k)}\in\BE_{\ga_k,m+2}\lpr{\tilde{Q}_\bs(A_r);\eps_r}$, and compute via \cite[Lemma 53]{GSLW} the block-encoding of the product $U_{T(\bs,\bbA)}\in\BE_{\ga,m+4}\lpr{T(\bs,\bbA);\eps}$ (called $U_{\bs,\bbA}$ in \hyperref[sec:mqet alg]{Algorithm 2}). Finally, 
        $$U_{g(\bbA)}=\LCU_\eps\lpr{\lpr{T(\bs,\bbA)}_{\bs\in[[D]]^r},\bgb}\in\BE_{\ga\norm{\bgb}_1,m+4+d}\lpr{g(\bbA);\eps}.$$

        \begin{algorithm}
        \caption{MQET for several commuting Hermitian matrices}\label{alg:mqet cap}
        \begin{algorithmic}[1]
        \Require $A_k\in\Mat(N)$ Hermitian and $U_k\in\BE_{\ga_k,m}(A_k;\eps_k)$ for $0\le k\le r$; $\lbr{A_{k_1},A_{k_2}}=0$ for all $0\le k_1,k_2\le r$; $f:[-1,1]^{r+1}\to\bD$; $g\in\C\lbr{x_0,\dots,x_r}$ with degree less than $D$ such that $\norm{f-g}_\infty$ is small
        \Ensure  $U_{f(\bbA)} \in \BE_{\ga',m'}(f(\bbA),\eps')$
        \For{$\bs\in[[D]]^r$}
            \State $\lpr{s_0,\dots,s_{r-1}}\gets\bs$
            \State $P_\bs^{(k)} \gets T_{s_k}$ 
            \Comment{for $0\le k<r$}
            \State $Q_\bs(x_r)\gets\prod\limits_{0\le k<r}\frac{2-\delta_{s_k,0}}{\pi}\int_{-1}^1g(\bx)\prod\limits_{0\le k<r}T_{s_k}(x_k)\frac{\d x_k}{\sqrt{1-x_k^2}}$
            \State \Comment{These satisfy $g(\bx)=\sum\limits_{\bs\in[[D]]^r}Q_\bs(x_r)\prod\limits_{0\le k<r}P_{k}(x_k)$}
            \State $\gb_\bs \gets \norm{Q_\bs}_\infty\prod\limits_{0\le k<r}\norm{P_\bs^{(k)}}_\infty$
            \State $\tilde{P}_\bs^{(k)} \gets \frac{P_\bs^{(k)}}{\norm{P_\bs^{(k)}}_\infty}$
            \Comment for $0\le k<r$
            \State $\tilde{Q}_\bs \gets \frac{Q_\bs}{\norm{Q_\bs}_\infty}$
            \State Compute $U_{\tilde{P}_\bs^{(k)}(A_k)}\in\BE_{\ga_k,m+2}\lpr{\tilde{P}_\bs^{(k)}(A_k);\eps_k}$
            \Comment for $0\le k<r$, using QSP
            \State Compute $U_{\tilde{Q}_\bs(A_r)}\in\BE_{\ga_r,m+2}\lpr{\tilde{Q}_\bs(A_r);\eps_r}$ \Comment Using QSP
            \State Compute $U_{\bs,\bbA}\in\BE_{\ga,m+4}\lpr{Q_\bs(A_r)\prod\limits_{0\le k<r}P_\bs^{(k)}(A_k);\eps}$
        \EndFor 
        \State \Return $\LCU_\eps\lpr{\lpr{U_{\bs,\bbA}}_{\bs\in[[D]]^r},\bgb} $
        \end{algorithmic}
    \end{algorithm}
    
    \subsection{Analysis of the algorithm}
        Recall that $\ga$, $m$, and $\eps$ are defined in \eqref{eq:ga m eps}. In the previous section, we kept a running tally of the ancilla count; the final tally was $2m+4+d$. We also track precision---$\eps$---and subnormalization factor---$\ga\norm{\bgb}$---above. 
        
        We can track the number of applications of $U_{A_\ell}$ and $U_{A_\ell}^\dag$. Applying QSP for degree at-most-$D$ polynomials multiplies this quantity by by at most $2D$. Computing products and running LCU adds; the final LCU call is on $D^r$ unitaries. Thus there are at most $2(D-1)rD^r\in O(rD^{r+1})$ instances. 
    
    \subsection{Subnormalization factor}\label{sec:gen subnorm factor}
        By enforcing $P_\bs^{(k)}:=T_{s_k}$ (the $s_k$th Chebyshev polynomial) and, as in \propref{prop:good scaling}, 
        $$Q_\bs(x_r):=\prod\limits_{1\le k\le r}\frac{2-\gd_{s_k,0}}{\pi}\int_{-1}^1g(\bx)\prod\limits_{1\le k\le r}T_{s_k}(x_k)\frac{\d x_k}{\sqrt{1-x_k^2}},$$
        we see that $\norm{P_\bs^{(k)}}_\infty=1$ and $\norm{Q_\bs}_\infty\le2$. This leaves us with $\gb_\bs\le2$. (Also as before, again using \propref{prop:int against cheb as sum} we can efficiently compute $Q_\bs$ given $g$.) Thus, $\norm{\bgb}=\sum\limits_{\bs\in[[D]]^r}2^{z(\bs)}$ where $z(\bs)$ is the number of zeroes appearing in the vector $\bs$. This rewrites to $\sum\limits_{0\le\ell\le r}2^\ell\binom{r}{\ell}D^{r-\ell}=(D+2)^r$, the subnormalization factor asserted in the beginning. 

\section{Conclusions}
    We briefly recapitulate some of the main points from the work above. The core ideas for implementing MQET are that polynomials admit relatively cheap decompositions, from the perspective of block-encoding parameters (as well as computationally for QSP phase factors, as the Chebyshev polynomial phase factors are trivial to compute), and that we harness commutativity (and the fact that the polynomials are tacitly in {\em commuting} variables) to freely plug in the matrix inputs. Then, for QET with a normal matrix, we exploit the observation that the real and imaginary parts are simply two commuting Hermitian matrices and can accordingly be used for bivariate MQET. Computationally, the main difficult tasks given the present state of the art is to select an appropriate polynomial approximation and to compute QSP phase factors; in particular, if working with the Chebyshev basis for all but the final polynomial in MQET, then that last polynomial can be computed from its coefficients as {\em sums} rather than needing to deal with numerical precision in an integral. 
        
\section*{Acknowledgements}

This project came out of Lin Lin's Fall 2021 course ``Quantum Algorithms for Scientific Computation'' at UC Berkeley, and has benefited from conversations with Lin, Subhayan Roy Moulik, and Nikhil Srivastava. YB was supported by NSF grant DMS-1952939. TS was supposed by NSF grant DMS-1646385. ZS was supported by a NSF graduate research fellowship, grant numbers DGE-1752814/2146752.

\appendix

\section{Selecting polynomial approximations}\label{sec:poly approx}

    Consider the task of computing a low-degree uniform-precision-$\eps$ polynomial approximation $p$ to a given continuous function $f:\bS\to\C$ on the square $\bS:=\bS^2=[-1,1]^2$, a special case of the $d$-hypercube $\bS^d:=[-1,1]^d$. 
    
    First, some definitions. Let $\cP_n$ be the complex polynomials of degree at most $n$. Let $\dist_{[a,b]}(f,\cP_n):=\inf\limits_{p\in\cP_n}\norm{f-p}_{L^\infty[a,b]}$, for $f$ any function taking values on $[a,b]$. We say that the {\em $n$th Chebyshev points on $[a,b]$} are $\bT_n=\lpr{t_0^{(n)},\dots,t_n^{(n)}}\in[a,b]^{n+1}$ which attain $\min\limits_{t_\ell\in[a,b]}\max\limits_{x\in[a,b]}\prod\limits_{0\le\ell<n}\abs{x-t_\ell}$. 
    For $0\le k<n$ and $\bT=\lpr{t_1,\dots,t_n}$, define the interpolating polynomials 
    $$L_n^{(k)}(x,\bT):=\prod\limits_{\substack{0\le\ell<n\\\ell\neq k}}\frac{x-t_\ell}{t_k-t_\ell}.$$
    Let $\gw(f,E,\gd):=\sup\limits_{\substack{x,y\in E\\\abs{x,y}<\gd}}\abs{f(x)-f(y)}$, for any compact set $E\subset\R^d$ and any continuous $f:E\to\R$. 
    
    Consider the following result. 
    \begin{theorem}[{\cite[Lemma 2.2]{Plesniak}}]\label{thm:plesniak}
        Consider the domain $P:=\prod\limits_{\ell=1}^d[a_\ell,b_\ell]\subset\R^d$ and any continuous function $f:P\to\R$. Select positive integers $n_\ell$ for $1\le\ell\le d$ and compute $\bT_{n_\ell}$.\footnote{The published statement of this theorem elects to work with Fekete (Legendre) points, but Chebyshev points feature an exponentially better Lebesgue constants and so are advantageous in this setting. The proof is unchanged.} Define 
        \begin{equation}
            p(\bx):=\sum\limits_{\bbk\in\prod\limits_{1\le\ell\le d}[[n_\ell]]}f\lpr{t_{k_1}^{(n_1)},\dots,t_{k_d}^{(n_d)}}\prod\limits_{1\le\ell\le d}L_{n_\ell}^{k_\ell}(x_\ell,\bT_{n_\ell}).\label{eq:p}
        \end{equation}
        Then, we have that\footnote{The published statement of this theorem erroneously omits the supremum on the right-hand side.} 
        $$\sup\limits_{\bx\in P}\abs{f(\bx)-p(\bx)}\le\sup\limits_{\bx\in P}\sum\limits_{1\le\ell\le d}\sum\limits_{\bj\in\prod\limits_{1\le k<\ell}[[n_k]]}(\log n_\ell+2)\dist_{[a_\ell,b_\ell]}\lpr{f\lpr{t_{j_1}^{(n_1)},\dots,t_{j_{\ell-1}}^{(n_{\ell-1})},\cdot,x_{\ell+1},\dots,x_d},\cP_{n_\ell}}.$$
    \end{theorem}
    Notice that $\deg p=\sum\limits_{1\le\ell\le d}n_\ell$. We are interested in the case of $d=2$ and $a_1=a_2=-1$, $b_1=b_2=1$, and $\deg p=n_1+n_2$. It just so happens that:
    \begin{proposition}[classical]
        The $n$th Chebyshev points on $[-1,1]$ are $t_\ell^{(n)}=\cos\lpr{\frac{2\pi\ell}{n}}$. 
    \end{proposition}
    Thus, our task is greatly simplified, as the bound from \thmref{thm:plesniak} reduces to
    $$\sum\limits_{0\le j<n_1}(\log n_2+2)\dist_{[-1,1]}\lpr{f\lpr{t_j^{(n_1)},\cdot},\cP_{n_2}}$$
    which is in turn boundable via:
    \begin{theorem}[{Jackson's inequality, cf.\ \cite[Theorem 1.4]{Rivlin}}]
        Let $f\in\cC[-1,1]$ be continuous. Then, $$\dist_{[-1,1]}(f,\cP_n)\le6\gw\lpr{f,[-1,1],\frac{1}{n}}.$$
    \end{theorem}
    and its explicit, more general version
    \begin{theorem}[{Jackson's inequality, cf.\ \cite[Theorem 2.1]{Plesniak}}]
        Let $f\in\cC^k[a,b]$ and $0\le k<n$. Then, $$\dist_{[a,b]}(f,\cP_n)\le\frac{\lpr{\frac{\pi}{4}(b-a)}^k}{k!\binom{n+1}{k}}\norm{f^{(k)}}_{L^\infty[a,b]}.$$
    \end{theorem}
    Notice that the denominator grows as $n^k$ for $n\gg k$, so if $\norm{f^{(k)}}_{L^\infty[a,b]}\ll n^k$ then the bound is interesting. 
    
    Therefore we can apply the bound to $\re f$ and $\im f$, supposing that each has the appropriate differentiability along each ``slice.'' Define the quantity $M_{\mathrm{r},k}(f):=\sup\limits_{t\in[-1,1]}\norm{\partial_2^k(\re f)(t,\cdot)}_{L^\infty[-1,1]}$ and similarly $M_{\mathrm{i},k}(f)$ for the imaginary part. Let $M_k(f):=\max\{M_{\mathrm{r},k}(f),M_{\mathrm{i},k}(f)\}$. 
    
    We thus have in full: 
    
    \begin{theorem}\label{thm:big 2var approx result}
        For all $n_1,n_2\in\N$, let $f\in\cC[-1,1]^2$ satisfy $\re f(t,\cdot),\im f(t,\cdot)\in\cC^k[-1,1]$ for all $t\in[-1,1]$ and some $0\le k\le n_1,n_2$. Then $p$ as in \eqref{eq:p} satisfies 
        $$\sup\limits_{\bx\in[-1,1]^2}\abs{f(\bx)-p(\bx)}\le\frac{\lpr{\nicefrac{\pi}{2}}^kM_k(f)}{k!}\frac{n_1\log n_2}{\binom{n_2+1}{k}}\in\tilde{O}\lpr{n^{1-k}}$$
        for $n_1,n_2\approx\nicefrac{n}{2}$. 
    \end{theorem}
    
    Of course, this is worst-case behavior, so in practice $f$ will converge to its best polynomial approximant even faster. 
    
    While it is often the case that we focus on functions defined only on $\bD\subset\bS\subset\C$, in our setting we assume that any such function of focus extends continuously to all of $\bS$ in such a way that allows us to apply \thmref{thm:big 2var approx result}, i.e.\ that it extends with $k$-times-differentiable real and imaginary parts. We name this the following: 
    \begin{property}[$(\g,k)$-extensibility]\label{property:extensibility}
        A function $f:\bD\to\C$ is said to be {\em $(\g,k)$-extensible} if there exists an extension $\tilde{f}:\bS\to\C$ satisfying $\tilde{f}=f$ on $\bD$, $\norm{\tilde{f}}_{L^\infty\bS}\le\g$, and $\re\tilde{f},\im\tilde{f}\in\cC^k[-1,1]$. 
    \end{property}
    
    \thmref{thm:big 2var approx result} has a natural version for $\bS^d$. For $f\in\cC\bS^d$, let 
    \begin{align*}
        M_{\mathrm{r},k}(f)&:=\max\limits_{2\le\ell\le d}\sup\limits_{\bt\in\bS^{d-1}}\norm{\partial_\ell^k(\re f)(t_1,\dots,t_{\ell-1},\cdot,t_\ell,\dots,t_{d-1})}_{L^\infty[-1,1]},\\
        M_{\mathrm{i},k}(f)&:=\max\limits_{2\le\ell\le d}\sup\limits_{\bt\in\bS^{d-1}}\norm{\partial_\ell^k(\im f)(t_1,\dots,t_{\ell-1},\cdot,t_\ell,\dots,t_{d-1})}_{L^\infty[-1,1]},\\
    \end{align*}
    and $M_k(f):=\max\lcr{M_{\mathrm{r},k},M_{\mathrm{i},k}}$. The same reasoning as before gets:
    \begin{theorem}
        For all $n_1,\dots,n_d\in\N$, let $f\in\cC[-1,1]^2$ satisfy 
        $$\re f(t_1,\dots,t_{\ell-1},\cdot,t_\ell,\dots,t_{d-1}),\im f(t_1,\dots,t_{\ell-1},\cdot,t_\ell,\dots,t_{d-1})\in\cC^k[-1,1]$$
        for all $\bt\in\bS^{d-1}$ and $2\le\ell\le d$ and some $0\le k\le\min\limits_{2\le\ell\le d}n_\ell$. Then $p$ as in \eqref{eq:p} satisfies 
        $$\sup\limits_{\bx\in\bS^d}\abs{f(\bx)-p(\bx)}\le\frac{\lpr{\nicefrac{\pi}{2}}^kM_k(f)}{k!}\sum\limits_{2\le\ell\le d}\frac{\log n_\ell}{\binom{n_\ell+1}{k}}\prod\limits_{1\le j<\ell}n_j.$$
    \end{theorem}

\section{Exponentiation of normal matrices}\label{sec:hamsim}
    
    The core observation---that normal $M$ decomposes into two attainable communiting Hermitians---can be used also for `Hamiltonian' simulation, where the goal is to obtain access to a block-encoding of $e^{Mt}$ for $t\in\R$, though as $M$ is not actually Hermitian the resulting exponential is itself not unitary. Nonetheless, it is plausible to desire such a matrix function, and the form \eqref{eq:split} readily resolves this task: 
    \begin{equation}\label{expsplit}
        e^{Mt}=e^{At}e^{iBt}
    \end{equation}
    precisely because $A$ and $B$ (or really, $At$ and $iBt$) commute. $e^{At}$ is unitary and achievable by standard methods (e.g.\ Trotterization, cf.\ \cite[Chapter 5]{Lin}). $e^{iBt}$ is not unitary but is approximable using a block-encoding of $B$ and then either QET as in \cite[\S7.5]{Lin}, by choosing a polynomial approximation to $x\mapsto e^{ixt}$, or as in \cite{TOSU 2}, by choosing a quadrature. It is then a matter of taking the product, adding the ancilla counts, and updating precisions accordingly (\propref{prop:be prod}).

\section{Optimizing the number of block-encoding instances}\label{sec:qetn circuit depth}

        QSP for block-encodings of Hermitian matrices entails using linear-in-polynomial-degree-many instances of the block-encoding, while our method is quadratic. We use relatively little in our method specific to normalcy, after using the existence of (commuting) $A=\re M$ and $B=\im M$. One might hope to find a modification where this count is $o(D^2)$ without sacrificing too many ancillae. Unfortunately, this turns out to not be possible with any similar scheme: 
        
        \begin{proposition}
            Let $G\in\C[x,y]$ have degree exactly $D-1$ and say $g(x+iy):=G(x,y)$. Suppose $\gT(D)$ of $G$'s ``maximal monomials,'' i.e.\ those of the form $x^\ell y^m$, $\ell+m=D-1$, have nonzero coefficients. Then $G$ has no decomposition as
            \begin{equation}
                g(x+iy)=G(x,y)=\sum\limits_{1\le k\le E}P_k(x)Q_k(y) \label{eq:generalized}
            \end{equation}
            where $P_k,Q_k\in\C[z]$ for $1\le k\le E$, and $E\in o(D)$. 
        \end{proposition}
        For clarity, we assume that in fact all $D$ of $G$'s maximal monomials have nonzero coefficient; however, the result certainly holds for any constant fraction of $D$. 

        Notice that this result holds no matter how large the degrees of $P_k$ and $Q_k$ are allowed to be. 
        \begin{proof}
            Write 
            $$G(x,y)=\sum\limits_{\substack{\ell,m\ge0\\\ell+m<D}}c_{\ell,m}x^\ell y^m$$
            and suppose that such $P_k(x)=\sum\limits_{0\le j<D'}a_{k,j}x^j$ and $Q_k(y)=\sum\limits_{0\le j<D'}b_{k,j}y^j$ as in \eqref{eq:generalized} exist, with $\deg P_k,\deg Q_k<D'$ (any $D'\in\N$ is possible here) for all $1\le k\le E$. The $x^\ell y^m$ coefficient of \eqref{eq:generalized}'s right-hand side is \begin{equation}
                \sum\limits_{1\le k\le E}a_{k,\ell}b_{k,m}=c_{j,\ell}. \label{eq:dot}
            \end{equation}
            This is just a dog product! So we are inspired to consider the following matrices: 
            \begin{align*}
                \cA&:=\lpr{a_{k,j}}_{\substack{1\le k\le E\\0\le j<D'}}=\begin{pmatrix} a_{1,0} & a_{2,0} & \cdots & a_{E,0} \\ a_{1,1} & a_{2,1} & \cdots & a_{E,1} \\ \vdots & \vdots & \ddots & \vdots \\ a_{1,D'} & a_{2,D'} & \cdots & a_{E,D'} \end{pmatrix} \\ 
                \cB&:=\lpr{b_{k,j}}_{\substack{1\le k\le E\\0\le j<D'}}=\begin{pmatrix} b_{1,0} & b_{2,0} & \cdots & b_{E,0} \\ b_{1,1} & b_{2,1} & \cdots & b_{E,1} \\ \vdots & \vdots & \ddots & \vdots \\ b_{1,D'} & b_{2,D'} & \cdots & b_{E,D'} \end{pmatrix} \\
                \cC&:=\lpr{c_{\ell,m}}_{\substack{0\le\ell<D\\0\le m<D}}=\begin{pmatrix} c_{0,0} & c_{0,1} & \cdots & c_{0,D-1} \\ c_{1,0} & \iddots & \iddots \\ \vdots & \iddots \\ c_{D-1,0} \end{pmatrix}
            \end{align*}
            where in $\cC$, the unspecified entries are all 0. Then, by \eqref{eq:dot}, $\cA\cB^\top=\cC$. We have that $E\ll D$ so $\rank\cA,\rank\cB\le E$, thus $\rank\lpr{\cA\cB^\top}\le E\ll D$. 
            
            However, $\det\cC=(-1)^{\floor{\nicefrac{n}{2}}}\prod\limits_{0\le k<D}c_{k,D-1-k}\neq0$ by hypothesis. Therefore $\cC$ is full-rank, i.e.\ $\rank\cC=D-1$. This stands in contradiction to $\rank\cC=\rank\lpr{\cA\cB^\top}\in o(D)$. 
        \end{proof}
        We interpret this as follows. The sum $\sum\limits_{0\le k\le E}\lpr{\deg P_k+\deg Q_k}$ is, up to a small constant factor, the number of calls to the block-encoding of the given normal matrix for implementing the decomposition \eqref{eq:generalized} using the QSP/LCU approach described herein. If we choose to approximate $f$ by $p(z,\ol{z})\in\C[z,\ol{z}]$, then $p$ is characterized by $O(D^2)$ complex numbers (one for each term $z^\ell\ol{z}^m$ for $\ell+m<D$), so we end up using about $D$-many polynomials, many or possibly all of which have degree $\gT(D)$. 
        
\section{Chebyshev expansion of Chebyshev interpolants}\label{sec:cheb exp cheb int}

In \secref{sec:subnorm} and \secref{sec:gen subnorm factor} it becomes necessary to evaluate expressions of the form 
$$\int_{-1}^1T_j(x)L_n^{(m)}(x)\frac{\d x}{\sqrt{1-x^2}}.$$
Here, we turn this integral into a sum of $n-1$ terms. 

First, we compute the Chebyshev coefficients for monomials: 
\begin{lemma}\label{lem:cheb monomial ip}
Take integers $m,j\ge0$. Suppose $m\ge j$ and $m$ and $j$ share the same parity. Then, 
$$\int_{-1}^1T_j(x)x^m\frac{\d x}{\sqrt{1-x^2}}=\frac{\pi}{2^m}\binom{m}{\frac{m-j}{2}}.$$
Otherwise, the integral evaluates to 0. 
\end{lemma}
\begin{proof}
The ``otherwise'' case is trivial by parity considerations and since the Chebyshev polynomials span $\R[x]$. 

For convenience, define $T_{-n}:=T_n$ for $n>0$ (in contrast to the convention $T_n\equiv0$ sometimes used). The result then follows by iteratively multiplying $T_n$ by $2x$: $2xT_n(x)=T_{n+1}(x)+T_{n-1}(x)$, which holds regardless of $n$'s sign. In particular, $(2x)^aT_b(x)=\sum\limits_{0\le c<a}\binom{a}{c}T_{a+b-2c}(x)$. 
\end{proof}

\begin{remark}
It is a standard fact  that 
\begin{equation}
    T_n(x)=\frac{n}{2}\sum\limits_{k=0}^{\floor{\nicefrac{n}{2}}}(-1)^k\frac{2^{n-2k}}{n-k}\binom{n-k}{k}x^{n-2k}. \label{eq:explicit cheb}
\end{equation}
\end{remark}
\begin{proposition}\label{prop:int against cheb as sum}
    For the Chebyshev polynomial $T_n$, if $k$ shares $n$'s parity then let $a_{n,k}:=\frac{n}{2}(-1)^{\frac{n-k}{2}}\frac{2^kn}{n+k}\binom{\frac{n+k}{2}}{\frac{n-k}{2}}$, otherwise let $a_{n,k}:=0$. Put $f_j(y):=\sum\limits_{0\le\ell<j}a_\ell y^{j-\ell}$ (which satisfies $f_j(y)=yf_{j-1}(y)+a_j$ for $j>0$). Then, 
    $$\int_{-1}^1T_k(x)L_n^{(i)}(x)\frac{\d x}{\sqrt{1-x^2}}=\pi\sum\limits_{\substack{k\le j<n\\2\mid(j-k)}}\frac{f_j(t_{n,i})}{2^j}\binom{j}{\frac{j-k}{2}}.$$
\end{proposition}
The advantage of this formulation is that it gives the integral in terms of a finite sum, in particular requiring $O(n)$ computations (since $f_k(t_{n,i})$ and $\binom{k}{0}=1$ can be computed initially as the base case for recursion). 
\begin{proof}
    From \eqref{eq:explicit cheb} we perform polynomial long division of $x-t_{n,i}$ into $T_n(x)$, and then use \lemref{lem:cheb monomial ip}. 
\end{proof}

\section{Nonlinear transformation of complex amplitudes primitives}\label{sec:ntca defs}
    
    In the interest of self-containedness, we present here the definitions of the gates $W$, $\tilde{G}$, and $\tilde{G}'$ as used in \secref{sec:apps}. 
    
    In \figref{fig:gmf 1}, $U$ is the prepare oracle for the starting state, i.e.\ 
    $$U\ket{0}^{\otimes n}=\sum\limits_{0\le k<n}c_k\ket{k},$$
    and $S$ and $H$ are the standard 1-qubit gates (in $\PU(2)$)
    \begin{align*}
        S=\begin{pmatrix}
            1&\\&i
        \end{pmatrix}, && H=\begin{pmatrix}
            1&\phantom{-}1\\1&-1
        \end{pmatrix}.
    \end{align*}
    (The function of the $S$ in $W'$ but not $W$ is to distinguish between the real and imaginary parts of the amplitudes.) The effect of $W$ is to establish 
    $$W\ket{k}\ket{0}^{\otimes n}\ket{0}=\ket{k}\lpr{\ket{\bv}\ket{+}+\ket{k}\ket{-}}$$
    (using the single-qubit states $\ket{\pm}=\frac{1}{\sqrt{2}}\lpr{\ket{0}\pm\ket{1}}$). 

\begin{figure}
    \centering
    \leavevmode
\Qcircuit @R=1em @C=1em {
& {/^n} \qw & \qw & \qw & \ctrl{1} & \qw & \qw & & & {/^n} \qw & \qw & \qw & \ctrl{1} & \qw & \qw & \qw \\
& {/^n} \qw & \gate{U} & \gate{U^\dag} & \targ & \qw & \qw & & & {/^n} \qw & \gate{U} & \gate{U^\dag} & \targ & \qw & \qw & \qw \\
& \qw & \gate{H} & \ctrl{-1} & \ctrl{-1} & \gate{H} & \qw & & & \qw & \gate{H} & \ctrl{-1} & \ctrl{-1} & \gate{S} & \gate{H} & \qw
    }
    \caption{The circuits for the gate $W$ (left) and $W'$ (right) \cite[Figure 1]{GMF}.}
    \label{fig:gmf 1}
\end{figure}
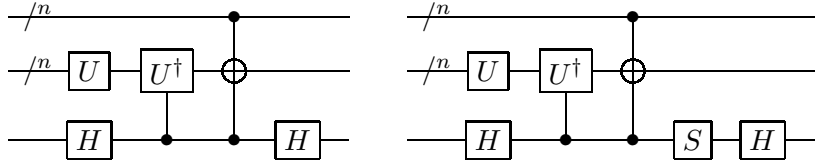

    In \figref{fig:gmf 2}, $Z$ refers to the standard 1-qubit gate (in $\PU(2)$) 
    $$Z=\begin{pmatrix}1\\&-1\end{pmatrix}$$ 
    and $S_0$ is reflection about the all-zeros state of the bottom two registers $\ket{0}^{\otimes n}\ket{0}$. The effect of $G$'s construction is to find an operator whose eigenstates include the explicitly-known eigenpairs
    $$\lpr{-\re c_k\pm i\sqrt{1-(\re c_k)^2},\frac{1}{\sqrt{2}}\ket{k}\lpr{\frac{1}{\ga_k}\lpr{\ket{\bv}+\ket{k}}\ket{0}\pm\frac{1}{\gb_k}\lpr{\ket{\bv}-\ket{k}}\ket{1}}}$$
    (for $\ga_k,\gb_k>0$ appropriate normalizing factors). 
    
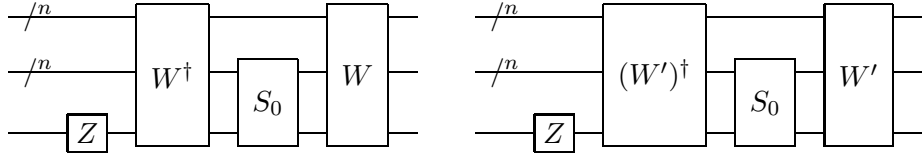
\begin{figure}
    \centering
    \leavevmode
\Qcircuit @R=1em @C=1em {
& {/^n} \qw & \qw & \multigate{2}{W^\dag} & \qw & \multigate{2}{W} & \qw & & & {/^n} \qw & \qw & \multigate{2}{(W')^\dag} & \qw & \multigate{2}{W'} & \qw \\
& {/^n} \qw & \qw & \ghost{W^\dag} & \multigate{1}{S_0} & \ghost{W} & \qw & && {/^n} \qw & \qw & \ghost{(W')^\dag} & \multigate{1}{S_0} & \ghost{W'} & \qw \\
& \qw & \gate{Z} & \ghost{W^\dag} & \ghost{S_0} & \ghost{W} & \qw & & & \qw & \gate{Z} & \ghost{(W')^\dag} & \ghost{S_0} & \ghost{W'} & \qw
    }
    \caption{The circuits for the gate $G$ (left) and $G'$ (right) \cite[Figure 2]{GMF}.}
    \label{fig:gmf 2}
\end{figure}

    Finally, in \figref{fig:gmf 3}, $X$ is the standard 1-qubit gate (in $\PU(2)$) 
    $$X=\begin{pmatrix}&1\\1\end{pmatrix};$$
    the circuit functions to give a LCU of $G$ and $G^\dag$ (or $G'$ and $(G')^\dag$), to eliminate any remaining imaginary parts. 

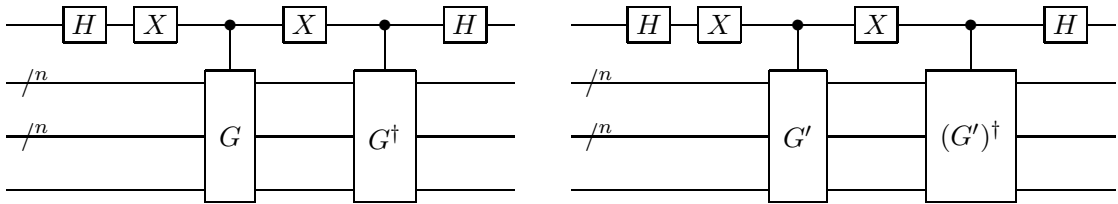
\begin{figure}
    \centering
    \leavevmode
\resizebox{\textwidth}{!}{
\Qcircuit @R=1em @C=1em {
& \qw & \gate{H} & \gate{X} & \ctrl{1} & \gate{X} & \ctrl{1} & \gate{H} & \qw & & & \qw & \gate{H} & \gate{X} & \ctrl{1} & \gate{X} & \ctrl{1} & \gate{H} & \qw \\
& {/^n} \qw & \qw & \qw & \multigate{2}{G} & \qw & \multigate{2}{G^\dag} & \qw & \qw & & & {/^n} \qw & \qw & \qw & \multigate{2}{G'} & \qw & \multigate{2}{(G')^\dag} & \qw & \qw \\
& {/^n} \qw & \qw & \qw & \ghost{G} & \qw & \ghost{G^\dag} & \qw & \qw & & & {/^n} \qw & \qw & \qw & \ghost{G'} & \qw & \ghost{(G')^\dag} & \qw & \qw \\
& \qw & \qw & \qw & \ghost{G} & \qw & \ghost{G^\dag} & \qw & \qw & & & \qw & \qw & \qw & \ghost{G'} & \qw & \ghost{(G')^\dag} & \qw & \qw
}
}
    \caption{The circuit for the gate $\tilde{G}$ (left) and $\tilde{G}'$ (right) \cite[Figure 3]{GMF}.}
    \label{fig:gmf 3}
\end{figure}

\end{document}